\documentclass[Svgc]{Svgc}
\usepackage{extraplaceins}
\usepackage{graphicx}
\usepackage{algorithm}
\usepackage{algpascal}
\usepackage{algpseudocode}
\algnotext{EndFor}
\algnotext{EndIf}
\journalname{Graphs and Combinatorics}

\algtext*{EndWhile}% Remove "end while" text
\algtext*{EndIf}% Remove "end if" text

\begin{document}
\title{Embedding a balanced binary tree on a bounded point set}
\author{Fatemeh Rajabi-Alni\inst{1}\thanks{\emph{Corresponding author:} fatemehrajabialni@yahoo.com} \and Alireza Bagheri\inst{2}% etc
% \thanks is optional - remove next line if not needed
%\thanks{\emph{Present address:} Insert the address here if needed}%
}                     % Do not remove

%running header
%\titlerunning{Running title}%
%\authorrunning{1st. author\inst{1} \and 2nd. author\inst{2}}%

%
%\offprints{}          % Insert a name or remove this line
%
\institute{Department of Computer Engineering, Islamic Azad University, North Tehran Branch, Tehran, Iran. \and Department of Computer Engineering and IT, Amirkabir University of Technology, Tehran, Iran.}
\maketitle
\begin{abstract}

Given an undirected planar graph $G$ with $n$ vertices and a set $S$ of $n$ points inside a simple polygon $P$, a \textit {point-set embedding} of $G$ on $S$ is a planar drawing of $G$ such that each vertex is mapped to a distinct point of $S$ and the edges are polygonal chains surrounded by $P$. A special case of the embedding problem is that in which $G$ is a balanced binary tree. In this paper, we present a new algorithm for embedding an $n$-vertex balanced binary tree $BBT$ on a set $S$ of $n$ points inside a simple $m$-gon $P$ in $O(m^2+n{\log^2 n}+mn)$ time with at most $O(m)$ bends per edge.
\end{abstract}
\receive{June, 2012}
\finalreceive{June 18, 2012}
\section{Introduction}
\label{IntroSect}

Let $G$ be an undirected $n$-vertex planar graph, and $S$ be a set of $n$ points. \textit {The point-set embedding problem}, also known as \textit {the point-set embeddability problem}, aims at drawing $G$ with no edges crossing, such that each vertex is mapped to a distinct point of $S$ and the edges are polygonal chains.

The largest subclass of planar graphs that admits a \textit {straight-line embedding}, i.e. an embedding with no edge-bends, on any point set is the class of outer planar graphs; this was shown by Gritzmann et al. \cite{11} for the first time, later it was rediscovered by Castaneda et al. \cite{7}, and finally Bose \cite{4} presented an efficient algorithm for it. The point-set embeddability problem has been extensively studied \cite{4,5,6,7,8,9,11,12,13,14,15,16}, but few researches have considered a bounding polygon for the point set \cite{3}. Note that deciding if there is a straight-line embedding of a general graph on a given set of points $S$ is NP-hard \cite {6}.

The constraint that the edges should be placed inside a surface with prescribed shape can be important for example to design a wired network on a surface with given shape.

Let $T$ be a tree with $n$ vertices and let $S$ be a set of $n$ points inside a polygon $P$ with $k$ reflex vertices, by an immediate implication of Lemma 7 of \cite {9} there exists an $O(n^2 \log n)$ time algorithm that computes a point-set embedding of $T$ on $S$ inside $P$ such that each edge has at most $2\left\lceil \frac{k}{2}\right\rceil $ bends.

Given a simple polygon $P$ with $m$ vertices, Bagheri et al. \cite {3} has presented an $O(m^2 n^{2.2})$ algorithm for planar poly-line drawing of an $n$-vertex complete binary tree on the surface of $P$, such that the total number of the edge bends is bounded by $O(m n^{1.6})$, in the worst case.

Let $BBT$ be an $n$-vertex balanced binary tree, i.e. a binary tree in which the depth of the two subtrees of every node differs by at most one. In this paper, we present a new algorithm for point-set embedding of $BBT$ on a set of $n$ points $S$ surrounded by a simple $m$-gon $P$, in $O(m^2+n{\log^2 n}+mn)$ time with at most $O(m)$ bends per edge. We use the partitioning idea of Bagheri et al. \cite {3} in our embedding algorithm.

This paper is organized as follows. Preliminary definitions are in Section \ref{PreliminSect}. In Section \ref{Newembedding}, we present our new embedding algorithm for embedding a balanced binary tree on a point set surrounded by a simple polygon. In Section \ref{exa}, we present an example for our new embedding algorithm. Conclusions and open problems are in Section \ref{Conclsection}.

\section{Preliminaries}
\label{PreliminSect}
\textit {The skeleton} of a simple polygons is a one-dimensional structure that describes the two-dimensional structure of the polygon with lower complexity than the original polygon.

There are two types of skeletons for a simple polygon, \textit {the medial axis} and \textit {the straight skeleton}. The medial axis of a simple polygon consists of all interior points of the simple polygon whose closest point on the boundary of it is not unique. The straight skeleton $SS$ is defined as the union of the pieces of angular bisectors traced out by the polygon vertices during a shrinking process \cite {1}. In this shrinking process all edges of $P$ are contracted inwards parallel to themselves and at a constant speed. So, each vertex of $P$ moves along the angular bisector of its incident edges until the boundary does not change topologically. Then, a new node is created and the shrinking process continues in the resulting polygon(s) until the area of all the resulting polygons becomes zero.

In a partition of a polygon into subpolygons, two subpolygons sharing an edge are called \textit{neighbours}. By Lemma \ref{face-num}, the straight skeleton of a simple polygon $P$ with $m$ vertices partitions the interior of $P$ into a cycle of $m$ monotone neighbor subpolygons, called \textit {faces}. According to Theorem $1$ of \cite{1}, each face of $P$ shares exactly one edge, called \textit {the boundary edge}, with the polygon $P$.

\begin{lemma} \cite{1}
\label{face-num}
The straight skeleton of a simple polygon $P$ with $m$ vertices consists of exactly $m$ connected faces, $m-2$ dummy vertices, and $2m-3$ dummy edges.
\end{lemma}

\begin{lemma} \cite{1}
\label{monotone}
Each face $f(e)$ is monotone in direction of its boundary edge $e$.
\end{lemma}

We denote the edge of $P$ that is adjacent to the face $f$ by $e(f)$, and the corresponding face of the boundary edge $e$ by $f(e)$.  
The bisector pieces of the straight skeleton of $P$ are called \textit {dummy edges}, and their endpoints that are not vertices of $P$ are called \textit {dummy vertices}.  We call the dummy edge that is connected to a boundary edge \textit {the boundary connected dummy edge} or $bcd$ edge. The dummy edges that are not $bcd$ edges are called \textit {the internal dummy edges}.

So, each face has one boundary edge and two $bcd$ edges. Also, it may have some internal dummy edges. Figure \ref{fig:1}(a) shows the straight skeleton of a simple polygon, in which the line segments $ab$ and $de$ are the two $bcd$ edges of the $f(ae)$, and $bc$ and $cd$ are the internal dummy edges. A corner $v$ of a polygon in the plane is said to be a \textit {reflex corner} if the angle at $v$ inside the polygon is greater than 180 degrees. In Figure \ref{fig:1}(b) the dashed line splits the reflex dummy vertex into two convex corners.

%figure 1

 \begin{figure}
\vspace{-4cm}
\hspace{-10cm}
\resizebox{2\textwidth}{!}{%
  \includegraphics{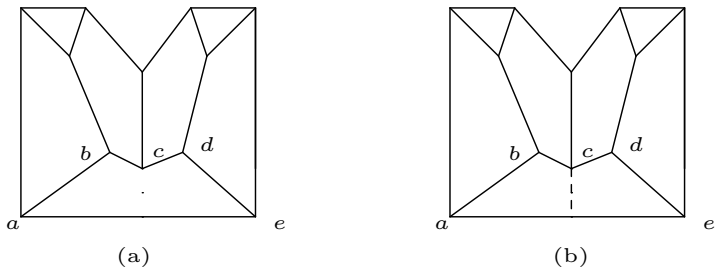}
}
% If not, use
%\vspace{5cm}       % Give the correct figure height in cm
\vspace*{-34.9cm}
\caption{The straight skeleton; definitions and notations.}
\label{fig:1}
\end{figure}

\textit {The neighbor vertices} of each vertex of a straight skeleton $SS$ are the vertices of $SS$ that are adjacent to it. The set of all the faces that are adjacent to a vertex of $SS$ are \textit {the neighbor faces} of it. In Figure \ref{fig:1}(a), the vertices $b$ and $d$ are the neighbor vertices of $c$, and the face $abcde$ is the neighbor face of the vertices $a$, $b$, $c$, $d$, and $e$. For each edge $(i,j)$ of $SS$ we define its weight by subtracting sum of the areas of adjacent faces of $i$ from the sum of the areas of adjacent faces of $j$. \textit {The middle edge} of $SS$ is the edge that has the minimum positive weight among all the other edges of $SS$. \textit {The middle point} of $SS$ is the middle point of its middle edge, if the middle edge is unique, otherwise the common point of the middle edges \cite{3}.

Now, we review the drawing algorithm presented by Bagheri et al. \cite{3}. Let $CBT$ be an $n$-vertex complete binary tree and $P$ be a simple polygon with $m$ vertices, their algorithm draws $CBT$ on the surface of $P$ in $O(m^2 n^{2.2})$ time with $O(m n^{1.6})$ total edge bends. This algorithm partitions the surface of the given simple $m$-gon into $\frac{n+1}{2}$ equal area subpolygons (where $n$ is the number of the nodes of the input complete binary tree) as follows. 

\textit{The bisecting chain} is a polygonal chain partitioning the input simple polygon into two equal area subpolygons $R$ and $L$. Their algorithm first computes a bisecting chain $C$ for the input simple polygon by a nonrecursive bisecting algorithm using the straight skeleton of $P$. The root of the input complete binary tree $CBT$ is laid on $C$. Then, the bisecting chains of the subpolygons $R$ and $L$ are computed such that each bisecting chain has one end point on $C$. This process is recursively repeated for drawing the left and right subtrees of $CBT$ inside $L$ and $R$, respectively. Bisecting the new subpolygons continues until getting $\frac{n+1}{2}$ new subpolygons. In fact, they lie the edges and internal nodes of the $CBT$ on the partitioning line segments and each leaf on the surface of one of the subpolygons.

\section{The new embedding algorithm}
\label{Newembedding}

In this section, we describe our new embedding algorithm which embeds an $n$-vertex balanced binary tree $BBT$ on a point set $S$ surrounded by a simple polygon $P$ with $m$ vertices. We use the idea of Bagheri et al. \cite{3}; using the straight skeleton of $P$ we decompose the surface of $P$ into convex regions such that each region contains zero or one single point of $S$ and the remaining points lie on the partitioning line segments. We use the partitioning chains as the edges.

%figure 5
\begin{figure}[h]
\vspace*{-5cm}
\hspace{-10cm}
\resizebox{2\textwidth}{!}{%
  \includegraphics{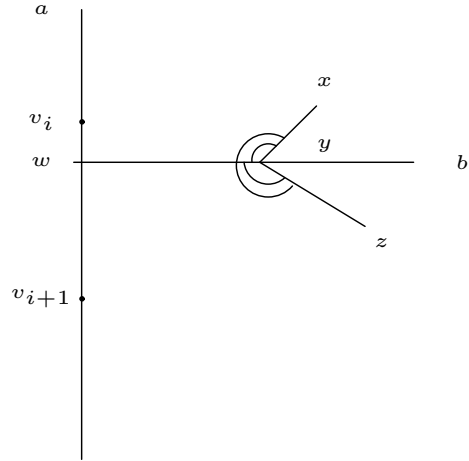}
}
\vspace*{-32cm}
\caption{Splitting a reflex vertex of a face into two convex vertices.}
\label{fig:5}       % Give a unique label
\end{figure}

%figure 6
\begin{figure}[h]
\vspace*{-6cm}
\hspace{-10cm}
\resizebox{2\textwidth}{!}{%
  \includegraphics{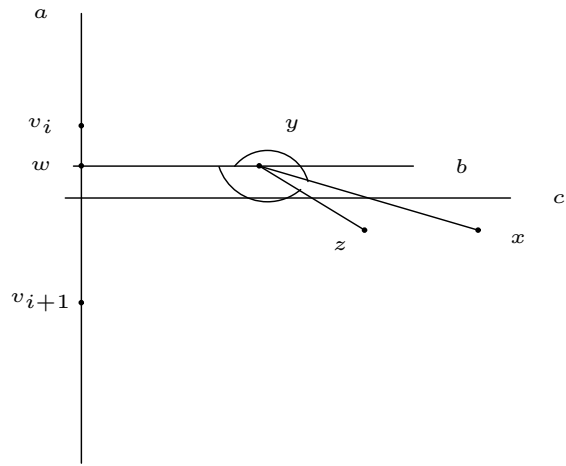}
}
\vspace*{-32cm}
\caption{Case I: $\widehat{xyw}>180$ degrees.}
\label{fig:6}       % Give a unique label
\end{figure}

%figure 7
\begin{figure}[h]
\vspace*{-5cm}
\hspace{-10cm}
\resizebox{2\textwidth}{!}{%
  \includegraphics{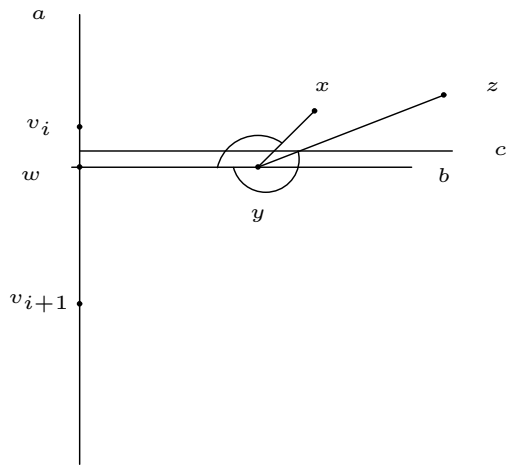}
}
\vspace*{-32cm}
\caption{Case II: $\widehat{wyz}>180$ degrees.}
\label{fig:7}       % Give a unique label
\end{figure}

\begin{lemma}
\label{lem:1}
Let $a$ be a straight line that passes through the line segment $e(f)$. Each line perpendicular to the straight line $a$ passing through a reflex vertex of $f(e)$, splits the reflex vertex into two convex vertices.
\end{lemma}

\begin{proof}
In Figure \ref{fig:5} we pass through the reflex vertex $\widehat{xyz}$ a straight line $b$ prependicular to the line $a$, where $x$, $y$ and $z$ are three consecutive vertices of $f(v_i v_{i+1})$. Assume that the lemma is false. Let $b$ be a line passing through the reflex vertex $\widehat{xyz}$ and perpendicular to the straight line $a$, and let $w$ be the intersection point of the lines $a$ and $b$. The line $b$ splits $\widehat{xyz}$ vertex into the vertices $\widehat{xyw}$ and $\widehat{wyz}$. Two cases arise: $\widehat{xyw}>180$ degree and $\widehat{wyz}>180$ degree. Figure \ref{fig:6} shows the first case, $\widehat{xyw}>180$. In Figure \ref{fig:6} there exists a line $c$ perpendicular to the straight line $a$ that intersects $f(v_i v_{i+1})$ in more than two points and this contradicts Lemma \ref{monotone}.
Figure \ref{fig:7} shows the second case, $\widehat{wyz}>180$ degrees. In Figure \ref{fig:7} there is a line $c$ perpendicular to the straight line $a$ that intersects $f(v_i v_{i+1})$ in more than two points and this contradicts Lemma \ref{monotone}.
\end{proof}

In this paper, the functions $Num(BBT)$ returns the number of the nodes of the balanced binary tree $BBT$. Also, the function $Num(P)$ returns the number of the points inside the simple polygon $P$. Let $BBT_1$ and $BBT_2$ denote the left and right subtrees of $BBT$, respectively. Our algorithm consists of two general steps, the recursive and the non recursive step (see Algorithm  \ref{EmbeddingAlg}). In the first step, the non recursive step, we compute a partitioning chain $C$ that divides the input simple polygon into two subpolygons $R$ and $L$, such that the number of the points in $L$ and $R$ are equal to $Num(BBT_1)$ and $Num(BBT_2)$, respectively and the remaining single point lies on $C$. Then in the second step, we recursively partition the new subpolygons $L$ and $R$ into other subpolygons until we get convex regions with zero or one single point in each region. If a point lie on a straight skeleton edge or on a decomposing edge, we assume that it is contained into an arbitrary face that is incident to it. In the following, we explain these steps.

\alglanguage{pseudocode}
%\alglanguage{pascal}
\begin{algorithm}[h]

\caption{The Embedding Algorithm}
\label{EmbeddingAlg}
\begin{algorithmic}[1]
\Statex \textbf{Input:} A simple polygon $P$, a balanced binary tree $BBT$, a point set $S$ inside $P$

\Statex \textbf{Output:} A point-set embedding of $BBT$ on $S$ such that the edges are polygonal chains inside $P$
\If{$BBT$ is  null}
\State Stop
\EndIf
\State $(R,L,SCR,SCL,q,SSSR,SSSL)=$\Call{PartitioningAlg}{P,S,BBT}
\State \Call{RecEmbedding}{L,SSSL,SCL,q,$BBT_1$,l}
\State \Call{RecEmbedding}{R,SSSR,SCR,q,$BBT_2$,r}
\State Map the root of $BBT$ on the point $q$

\end{algorithmic}
 \end{algorithm}

\noindent \textbf{Step 1. The non recursive step.}

In this step, the partitioning algorithm \textit{PartitioningAlg} is called (line $3$ of Algorithm \ref{EmbeddingAlg}). It partitions the surface of $P$ into two subpolygons $R$ and $L$ using the straight skeleton of $P$, such that $Num(BBT_1)=Num(L)$, $Num(BBT_2)=Num(R)$, and the remaining single point lies on the partitioning chain.

We first partition $P$ into a sequence of the neighbor convex subpolygons using the straight skeleton of $P$ as follows. We find the straight skeleton $SS$ of $P$. Then, we split each reflex vertex of the faces into the convex vertices. By Lemma \ref{lem:1}, we can split each reflex vertex of a face $f(e)$ into two convex vertices using a straight line passing through the reflex vertex and perpendicular to the straight line that contains the edge $e$. We find the middle point of $SS$, called $ms$, lying on the boundary of one or more faces of $SS$. We select one of the incident faces of $ms$ as the starting face $sf$. Let $\partial P$ be the boundary of $P$, considered as a counter clockwise path $CCW$, starting from and ending at the boundary edge of the starting face $sf$. Starting from $sf$, we process each face, one by one, in the order its boundary edge appears in $CCW$, and add the splitting line segments to the reflex vertices. So we pass a straight line segment, called \textit{splitting line segment}, through each reflex vertex of the faces of $P$. The union of $SS$ and the splitting line segments is a structure called \textit{split straight skeleton} or $SSS$. After adding the splitting line segments, all subfaces of $SSS$ are convex; in Figure \ref{fig:1}(b) the dashed line shows the splitting line segment. The function $Split(P)$ adds the splitting line segments to $SS$. Each face of $P$ is incident to one edge of $P$, so $SSS$ partitions the surface of $P$ into a cycle of the neighbor convex subfaces.

Then, we pass through $ms$ a straight line segment, called ${sl}_1$ that is perpendicular to $e(f(ms))$ (Figure \ref{fig:8}). By Lemma \ref{monotone} this line splits one of the subfaces of $f$ into at most two simple subfaces. We call the subface that is in the same direction with $e(f)$, $f(t)$ and the other subface, if exists, $f(m)$. 

%figure 8
\begin{figure}[h]
\vspace*{-5cm}
\hspace{-10cm}
\resizebox{2\textwidth}{!}{%
  \includegraphics{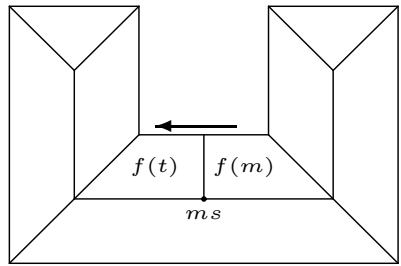}
}
\vspace*{-35cm}
\caption{Two new subfaces $f(m)$ and $f(t)$.}
\label{fig:8}       % Give a unique label
\end{figure}

Let $SC$ be a sequence of the consecutive subfaces of $SSS$, where the faces are in $CCW$ order and the subfaces of each face $f$ are ordered with increasing $e(f)$ direction such that $SC.first=f(t)$ and $SC.last=f(m)$ (see Figure \ref {fig:9}). Starting from $f(t)$ we process the subfaces, one by one, in the order that they appear in $SC$. Let $S_1=Num(BBT_1)$ denote the number of the nodes of the left subtree of $BBT$. Let $SUM$ denote the sum of the points surrounded by the processed subfaces. We process the subfaces until $SUM$ be greater than $S_1$ and then stop. Then, we split the current subface $cf$ into two subfaces $cf(m)$ and $cf(t)$ using a straight line segment, called \textit{the dividing line segment}. Note that $cf(m)$ is the subface that is not in the same direction with $e(cf)$ and the other subface is $cf(t)$. Moreover, the dividing line segment must intersect a point of $cf$ such that the sum of the points surrounded by the processed subfaces, i.e. $SUM$, plus $Num(cf(m))$ be equal to $S_1$. Each subface may have four types of the edges: a boundary edge, two boundary connected edges, the splitting line segments, and some dummy edges. We select a point on one of the dummy edges of the current subface, called \textit{center point}. If the current subface $cf$ has no dummy edges and so is a triangle, we select the intersection point of the two $bcd$ edges of it. In fact, we use the dividing line segments used in proof of Lemma 6 by Gritzmann et al. \cite{11}. Therefore, we order all the points surrounded by $cf$ radially around the center point, find $(S_1-SUM+1)th$ point inside $cf$, and index it by $q$. A dividing line segment ${sl}_2$ is a line segment passing through the center point and $q$. The dividing line segments exist since all subpolygons are convex. So we partition the simple polygon $P$ into two new subpolygons $R$ and $L$, such that $L$ is the subpolygon between ${sl}_1$ and ${sl}_2$ that contains $f(t)$ and $R$ is the subpolygon between ${sl}_1$ and ${sl}_2$ that contains $f(m)$ (Figure \ref{fig:10}). In fact, the subfaces form $f(t)$ to $cf(m)$ in the $CCW$ direction constitute $L$. And $R$ consists of the subfaces form $f(m)$ to $cf(t)$ in the clockwise direction.

Consider the boundary of the subpolygon $L$ as a clockwise path ${path}_1$ and the boundary of $R$ as a counter clockwise path ${path}_2$ starting from and ending at ${sl}_2$. Then, $SCL$ and $SCR$ are two lists of consecutive subfaces in ${path}_1$ and ${path}_2$ order, respectively. Figure \ref{fig:11} shows $SCL=f_1\rightarrow f_2\rightarrow f_3\rightarrow f_4\rightarrow f_5$ (the subfaces of $L$) and $SCR=f'_1\rightarrow f'_2\rightarrow f'_3\rightarrow f'_4\rightarrow f'_5$ (the subfaces of $R$). Let $SSS(R)$ and $SSS(L)$ be functions that return the portions of the $SSS$ that corresponds to $R$ and $L$, respectively. The pseudo-code of the non recursive algorithm is given in the Figure \ref{fig:15}.

%figure 9
\begin{figure}[h]
\vspace*{-5cm}
\hspace{-10cm}
\resizebox{2\textwidth}{!}{%
  \includegraphics{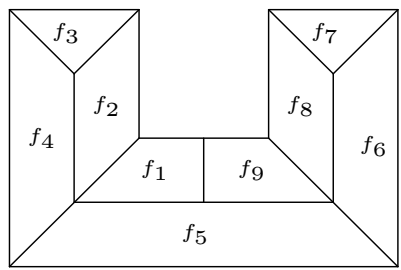}
}
\vspace*{-35cm}
\caption{A sequence of the neighbor faces.}
\label{fig:9}       % Give a unique label
\end{figure}

%figure 10
\begin{figure}[h]
\vspace*{-5cm}
\hspace{-10cm}
\resizebox{2\textwidth}{!}{%
  \includegraphics{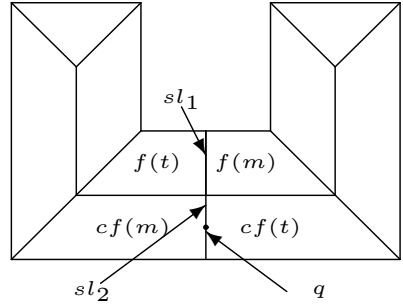}
}
\vspace*{-34cm}
\caption{Two new subpolygons $R$ and $L$.}
\label{fig:10}       % Give a unique label
\end{figure}

%figure 11
\begin{figure}[h]
\vspace*{-5cm}
\hspace{-10cm}
\resizebox{2\textwidth}{!}{%
  \includegraphics{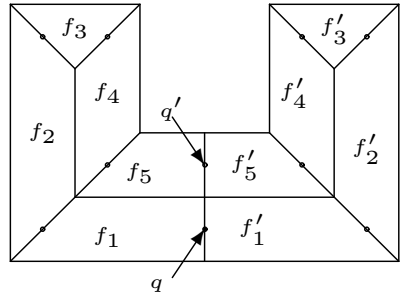}
}
\vspace*{-34.5cm}
\caption{The points $q$, $q'$, and the middle points of the common borders of the subfaces.}
\label{fig:11}       % Give a unique label
\end{figure}

%figure 15

\alglanguage{pseudocode}
%\alglanguage{pascal}
\algsetblock[Name]{Initial}{}{3}{1cm}
\begin{algorithm}[h]

\caption{The Partitioning Algorithm}
\label{PartitioningAlg}
\begin{algorithmic}[1]
\Statex \textbf{Input:} A simple polygon $P$, a point set $S$ inside $P$, a balanced binary tree $BBT$

\Statex \textbf{Output:}  Two simple polygons $R$ and $L$, two lists $SCR$ and $SCL$, the point $q\in S$ on the partitioning polygonal chain, two skeleton like structures $SLS(R)$ and $SLS(L)$

\State Set $SS=StraightSkel(P)$, $ms=MiddlePoint(SS)$, $SSS=Split(SS)$, $f=$ an incident face to $ms$, $SUM=0$, $BBT_1=LeftSubTree(BBT)$
\State Split $f$ into $f(m)$ and $f(t)$ using a line passing through $ms$

\State Create the list of the subfaces of $SSS$ called $SC$
\State $S_1=Num(BBT_1)$
\State $f=SC.first$ 
\While {$SUM+Num(f)\leq S_1$} 
\State $SUM=SUM+Num(f)$
\State $f=f.next$
\EndWhile
\State $cf=f$
\State Divide $cf$ into $cf(m)$ and $cf(t)$ using a line passing through the point $q$ inside $cf$
\State Construct the list of the left new subpolygon $L$ subfaces called $SCL$ \State Construct the list of the right new subpolygon $R$ subfaces called $SCR$ 
\State Return two new subpolygons $R$ and $L$, $SCR$ and $SCL$, the point $q$, $SSS(R)$ and $SSS(L)$.
\end{algorithmic}
 \end{algorithm}

\noindent \textbf{Step 2. The recursive step.}

Given the point $q$ and the new subpolygons $R$ and $L$ computed in the previous step, we call the recursive embedding algorithm, called $RecEmbeddingAlg$, for embedding the left and right subtrees of $BBT$ on the points inside $L$ and $R$, respectively (line $4$ and $5$ of Algorithm \ref{EmbeddingAlg}). 

The recursive embedding algorithm $RecEmbeddingAlg$ takes a simple polygon $P$ which consists of a sequence of the neighbor subfaces $SC$, the skeleton like structure $SLS$ of $P$, a point $q\in S$ on the border of the first subface of $SC$, and a balanced binary tree $BBT$. Moreover, the last input $sid$ indicates whether the input subpolygon is $R$ or $L$. The output of $RecEmbeddingAlg$ is a point set embedding of $BBT$ on $S$ inside $P$ (see Algorithm \ref{RecEmbeddingAlg}).

We have computed the straight skeleton in the previous step, while in this step we compute the straight skeleton like structure $SLS$ of the input simple polygon as follows. Given the sequence of the consecutive neighbor subfaces $SC$ and the point $q\in S$ on the border of the first subface of $SC$, we insert a dummy node at the middle point of the common border of each two neighbor subfaces of $SC$ called \textit{the backbone point}. Moreover, we insert a dummy node at the middle point of $$Perimeter\left(SC.last\right)-Commonedge\left(SC.last,\ SC.last.prev\right),$$ called $q'$. 

Consider the backbone points of the neighbor subfaces of $SC$, the starting point $q$, and the point $q'$ (Figure \ref{fig:11} depicts these points on the subfaces of $L$ and $R$). We connect these points to each other, using straight line segments, called \textit{backbone line segments}. Consider \textit{the backbone} as a directed path $qq'$ that partitions the consecutive subfaces into two sequences of the subfaces, the sequence of the subfaces that are in the left side of the backbone, called $SCL$ and the ones in the right side of it, called $SCR$. We slightly perturb the backbone line segments such that all new subfaces be convex, since finding a dividing line segment requires a convex simple polygon.

Now, we construct a new list by concatenating the lists $SCR$ and $SCL$. Two cases arise; $sid=r$ and $sid=l$. In the first case, that is the case in which the input subpolygon is $R$, let $SC= Concatenation(SCR,SCL)$ where $SC.first=SCR.first$ and $SC.last=SCL.last$. 
Now, assume that the input subpolygon is $L$, in this situation $SC= Concatenation(SCL,SCR)$ where $SC.first=SCL.first$ and $SC.last=SCR.last$. 

Then, we process the subfaces of $SC$, starting from $SC.first$, one by one (Figures \ref{fig:12} and \ref{fig:13} depict the consecutive subfaces of $L$ and $R$). As the non recursive step, when the sum of the points surrounded by the processed subfaces, $SUM$, is greater than the number of the nodes of the left child of $BBT$, $Num(BBT_1)$, the process stops. Then, we select a point of the backbone that is incident to the current subface $cf$, called the point $t'$, order all the points surrounded by $cf$ radially around $t'$, and find $(Num(BBT_1)-SUM+1)th$ point inside $cf$, called $t$. Then, we divide the current subface by the dividing line segment passing through $t'$ and $t$.
We map the root of $BBT$ on $t$, and connect $t$ to $q$ using $tt'$ and the portion of the backbone between $q$ and $t'$. And finally, we recursively call the recursive embedding algorithm for subtrees of $BBT$, and the two newly created subpolygons (see Figure \ref{fig:lr}).

\alglanguage{pseudocode}
%\alglanguage{pascal}
\begin{algorithm}

\caption{The Recursive Embedding Algorithm}
\label{RecEmbeddingAlg}
\begin{algorithmic}[1]
\Statex \textbf{Input:} A simple polygon $P$, its skeleton like structure $SLS$, a sequence of the consecutive subfaces $SC$, a balanced binary tree $BBT$, a point set $S$ inside $P$, an starting point $q\in S$ which is incident to $SC.first$, the direction $sid$ that in which the input subpolygon is placed 

\Statex \textbf{Output:} A point set embedding of $BBT$ on $S$ inside $P$

\If{$BBT$ is  null}
\State Stop
\EndIf
\State $SUM=0$
\State Set $BBT_1=LeftSubTree(BBT)$, and $BBT_2=RightSubTree(BBT)$
\State Construct the backbone of $SC$
\State Construct two new lists of the subfaces $SCR$ and $SCL$  
\If{$sid=r$}
\State $SC=Concatenation(SCR,SCL)$
\Else 
\State $SC=Concatenation(SCL,SCR)$
\EndIf
\State $f=SC.first$.
\While {$SUM+Num(f)\leq S_1$} 
\State $SUM=SUM+Num(f)$
\State $f=f.next$
\EndWhile
\State Set $cf=f$
\State Divide $cf$ into $cf(m)$ and $cf(t)$ using a line passing through the point $t$ inside $f$
\State  Construct the subface list of the left new subpolygon $L$ called $SCL$
\State  Construct the subface list of the right new subpolygon $R$ called $SCR$
\State Map the root of $BBT$ on the point $t$
\State Connect $q$ to $t$
\State \Call {RecEmbeddingAlg}{R,SLS(R),SCR,t,$BBT_1$,r}
\State \Call {RecEmbeddingAlg}{L,SLS(L),SCL,t,$BBT_2$,l}

\end{algorithmic}
 \end{algorithm}

%figure 12
\begin{figure}[h]
\vspace*{-4cm}
\hspace{-10cm}
\resizebox{2\textwidth}{!}{%
  \includegraphics{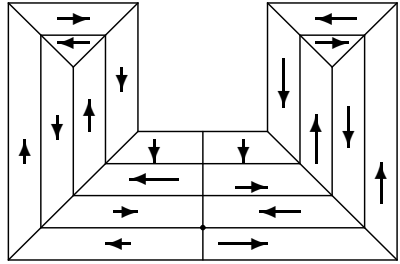}
}
\vspace*{-35cm}
\caption{The direction of processing the new subfaces.}
\label{fig:12}       % Give a unique label
\end{figure}

%figure 13
\begin{figure}[h]
\vspace*{-4cm}
\hspace{-10cm}
\resizebox{2\textwidth}{!}{%
  \includegraphics{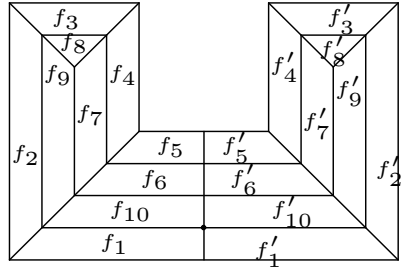}
}
\vspace*{-35cm}
\caption{The new subfaces.}
\label{fig:13}       % Give a unique label
\end{figure}

\begin{theorem}
\label{theorem1}
Given an $n$-vertex balanced binary tree $BBT$ and a set of $n$ points $S$ bounded by a simple $m$-gon $P$, we can embed $BBT$ on $S$ inside $P$ using our embedding algorithm in $O(m^2+n\log^2 n+mn)$ time, such that the bend number of each edge is bounded by $O(m)$ in the worst case.
\end{theorem}

\begin{proof}
In Step 1, we first compute the straight skeleton of the input simple $m$-gon using the algorithm of \cite{10} in $O(m^2)$ time, and then add the splitting line segments in $O(m)$ time. Processing the points inside the subfaces and finding the dividing line segment take $O(n+m)$ and $O(n \log n)$ time, respectively. Hence, Step 1 runs in $O(m^2+n \log n)$ time.

In Step $2$, we recursively solve two subproblems, each of size $O(n/2)$, which contributes $2T(n/2)$ to the running time. The time complexity of the recursive embedding algorithm, $T\left(n\right)$, is bounded by $O(n\log^2 n+mn)$ by solving the recursive equation $$T\left(n\right)=2T\left(\frac{n}{2}\right)+n{\log n}+O(m).$$ Where $n{\log n}$ is the time needed for sorting $n$ points inside the simple polygon and finding the dividing line segment. Moreover, $O(m)$ is the time of constructing the new skeleton like structures of new subpolygons in each recursion. So, the time complexity of the embedding algorithm is $O(m^2+n\log^2 n+mn)$. 

The bend number of each edge is bounded by $O(m)$ in the worst case, since each edge of the input balanced binary tree is a portion of one of the straight skeleton like structures $SLS$. Recall that the $SLS$ structures are constructed by connecting the middle points of the $bcd$ edges and the splitting edges, whose total number is bounded by $O(m)$, to each other.
\end{proof}

\section{An Example}
\label{exa}

In this section, we present an example for point-set embedding of a $15$-vertex complete binary tree on the point set of Figure \ref{fig:17} using our new algorithm. In Figures \ref{fig:18}, \ref{fig:19}, and \ref{fig:20}, the vectors depict the paths used as the edges from the circled nodes to their children. Figure \ref{fig:21} represents the output of our embedding algorithm after removing the additional line segments.

%figure 17
\begin{figure}[h]
\vspace*{-4cm}
\hspace{-10cm}
\resizebox{2\textwidth}{!}{%
  \includegraphics{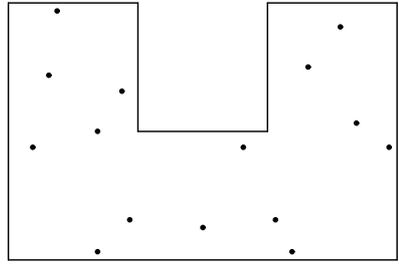}
}
\vspace*{-35cm}
\caption{A point set surrounded by a simple polygon.}
\label{fig:17}       % Give a unique label
\end{figure}

%figure 18
\begin{figure}[h]
\vspace*{-5cm}
\hspace{-10cm}
\resizebox{2\textwidth}{!}{%
  \includegraphics{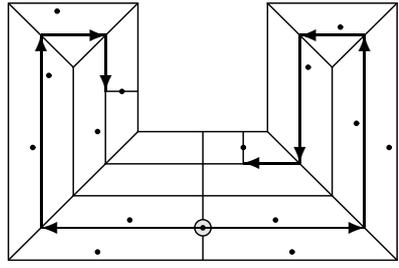}
}
\vspace*{-35cm}
\caption{Connecting the root of the $15$-vertex complete binary tree to its children.}
\label{fig:18}       % Give a unique label
\end{figure}

%figure 18
\begin{figure}[h]
\vspace*{-5cm}
\hspace{-10cm}
\resizebox{2\textwidth}{!}{%
  \includegraphics{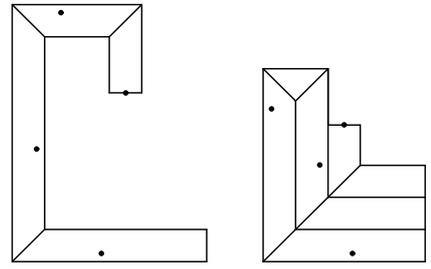}
}
\vspace*{-35cm}
\caption{Two new subpolygons of the left subpolygon of Figure \ref{fig:18} }
\label{fig:lr}       % Give a unique label
\end{figure}

%figure 19
\begin{figure}[h]
\vspace*{-5cm}
\hspace{-10cm}
\resizebox{2\textwidth}{!}{%
  \includegraphics{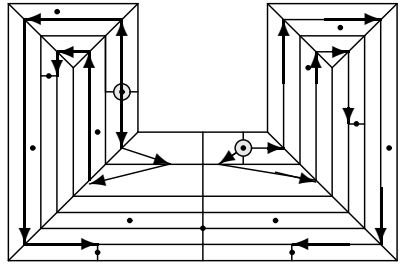}
}
\vspace*{-35cm}
\caption{The edges connecting the nodes of the second level to their children.}
\label{fig:19}       % Give a unique label
\end{figure}

%figure 20
\begin{figure}[h]
\vspace*{-6cm}
\hspace{-10cm}
\resizebox{2\textwidth}{!}{%
  \includegraphics{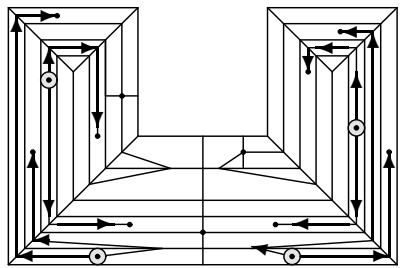}
}
\vspace*{-34.5cm}
\caption{The edges that connect the leaves to their parents.}
\label{fig:20}       % Give a unique label
\end{figure}

%figure 21
\begin{figure}[h]
\vspace*{-5cm}
\hspace{-10cm}
\resizebox{2\textwidth}{!}{%
  \includegraphics{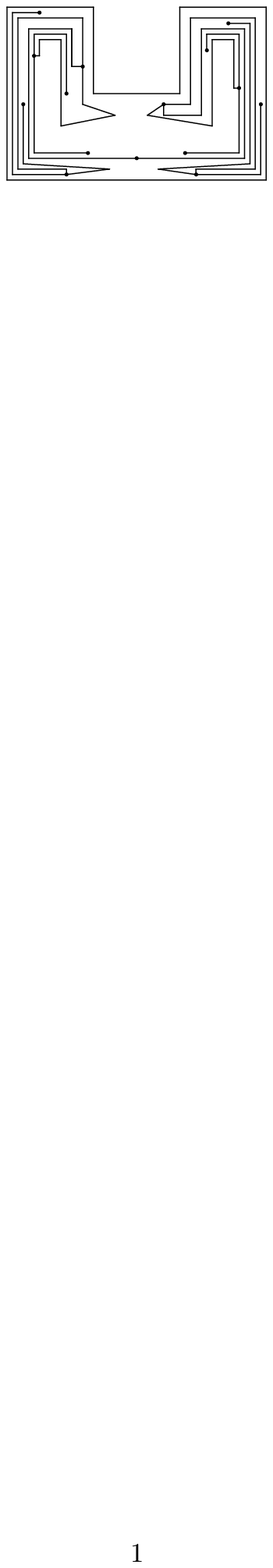}
}
\vspace*{-35cm}
\caption{The embedded $15$-vertex complete binary tree on the point set of Figure \ref{fig:17} by our embedding algorithm.}
\label{fig:21}       % Give a unique label
\end{figure}

\section{ Concluding Remarks}
 \label{Conclsection}

In this paper, we introduce a new algorithm for point set embedding of an $n$-vertex balanced binary tree $BBT$ on a set of $n$ points $S$ bounded by a simple polygon $P$ with $m$ vertices. Our new embedding algorithm computes the output in $O(m^2 + n\log^2 n+mn)$ time with at most $O(m)$ bends per edge. As a future work, we can concentrate on the other restricted types of the graphs and polygons.

%\begin{acknowledgement}

%\end{acknowledgement}
\FloatBarrier

\end{document}